\documentclass[a4paper,twocolumn,english,showpacs,nofootinbib,nobibnotes,aps,prl,floatfix,superscriptaddress]{revtex4-2}
\usepackage[utf8]{inputenc}
\usepackage{amsmath}
\usepackage{amsthm}
\usepackage{amssymb}
\usepackage{graphicx}
\usepackage{xcolor}
\usepackage{rotating}
\usepackage{mathtools}
\usepackage[caption=false]{subfig}
\usepackage{xparse}
\usepackage{tikz}
\usepackage{pgfplots}
\usetikzlibrary{plotmarks}
\pgfplotsset{compat=newest}

\usetikzlibrary{shapes.misc}
\newcounter{cox}
\newcounter{coy}
\newcommand{\sudoku}[3]{%
  \begin{tikzpicture}[darkstyle/.style={draw,fill=purple!67,minimum size=#1}, lightstyle/.style={draw,fill=purple!0,minimum size=#1}, nostyle/.style={fill,fill=white, line width=0.0mm,minimum size=#1}]
  \setcounter{cox}{0}
  \setcounter{coy}{1}
  \foreach \elem in {#3} {
    \edef\nlabel{\elem}
    \if \elem X
      \node [darkstyle]  (\elem) at (3.0/#2*\thecox,-3.0/#2*\thecoy) { }; 
    \else
      \node [lightstyle]  (\elem) at (3.0/#2*\thecox,-3.0/#2*\thecoy) { }; 
    \fi
    \stepcounter{cox}
    \pgfmathtruncatemacro\coxp{\thecox+1.0}
    \ifnum\coxp>#2
      \stepcounter{coy}
      \setcounter{cox}{0}
    \fi
  }
  \end{tikzpicture}
  
}
\newcommand{\sudokucolor}[5]{%
  \begin{tikzpicture}[darkstyle/.style={draw,fill=#4,minimum size=#1}, lightstyle/.style={draw,fill=#5,minimum size=#1}, nostyle/.style={fill,fill=white, line width=0.0mm,minimum size=#1}]
  \setcounter{cox}{0}
  \setcounter{coy}{1}
  \foreach \elem in {#3} {
    \edef\nlabel{\elem}
    \if \elem X
      \node [darkstyle]  (\elem) at (1.78/#2*\thecox,-1.78/#2*\thecoy) { }; 
    \else
      \node [lightstyle]  (\elem) at (1.78/#2*\thecox,-1.78/#2*\thecoy) { }; 
    \fi
    \stepcounter{cox}
    \pgfmathtruncatemacro\coxp{\thecox+1.0}
    \ifnum\coxp>#2
      \stepcounter{coy}
      \setcounter{cox}{0}
    \fi
  }
  \end{tikzpicture}
  
}

\newcommand{\maxcolor}{purple}

\newcommand{\wernerzero}{
\sudokucolor{24}{2}{
O,O,
O,X}{\maxcolor!100}{\maxcolor!0}
}
\newcommand{\wernerone}{
\sudokucolor{24}{2}{
O,O,
O,X}{\maxcolor!67}{\maxcolor!11}
}
\newcommand{\wernertwo}{
\sudokucolor{24}{2}{
O,O,
O,X}{\maxcolor!33}{\maxcolor!22}
}
\newcommand{\wernerthree}{
\sudokucolor{24}{2}{
O,O,
O,X}{\maxcolor!0}{\maxcolor!33}
}

\newcommand{\patternfour}{
\sudoku{20.25}{4}{
X,O,O,O,
O,X,X,X,
O,O,X,O,
O,O,X,O}
}

\newcommand{\patternthree}{
\sudoku{27}{3}{
X,O,O,
O,X,X,
O,X,X}
}
\newcommand{\patternthreetwo}{
\sudoku{27}{3}{
X,O,O,
O,X,O,
O,O,X}
}

\newcommand{\patternsix}{
\sudoku{13.5}{6}{
X,X,X,X,X,.,
X,X,.,.,.,.,
.,X,.,.,.,X,
X,.,.,X,.,.,
.,.,.,X,.,X,
X,.,.,.,.,X}
}

\newcommand{\patternfoursix}{
\sudoku{13.5}{6}{
X,.,X,.,.,.,
X,X,X,.,X,.,
X,.,X,.,.,.,
.,.,.,X,.,X}
}

\newcommand{\patternfoursixtwo}{
\sudoku{13.5}{6}{
X,.,.,.,X,.,
X,X,X,X,X,.,
X,.,.,.,X,.,
.,.,.,.,.,X}
}

\newcommand{\patternthreefour}{
\sudoku{20.25}{4}{
X,.,.,.,
X,X,X,X,
.,.,.,X}
}

\newcommand{\patternthreefourtwo}{
\sudoku{20.25}{4}{
.,.,.,X,
X,X,X,X,
X,.,.,.}
}

\newcount\colveccount
\newcommand*\colvec[1]{
        \global\colveccount#1
        \begin{pmatrix}
        \colvecnext
}
\def\colvecnext#1{
        #1
        \global\advance\colveccount-1
        \ifnum\colveccount>0
                \\
                \expandafter\colvecnext
        \else
                \end{pmatrix}
        \fi
}

\makeatletter

\newcounter{thmc}
\newtheorem{proposition}[thmc]{Proposition}

\newtheorem{theorem}[thmc]{Theorem}

\newtheorem{lemma}[thmc]{Lemma}

\newtheorem{definition}[thmc]{Definition}
\newcommand{\pushright}[1]{\ifmeasuring@#1\else\omit\hfill$\displaystyle#1$\fi\ignorespaces}

\usepackage{braket}
\usepackage{dsfont}

\makeatother

\usepackage{babel}

\global\long\def\Ad{\operatorname{Ad}}
\global\long\def\trace{\operatorname{Tr}}
\global\long\def\ketbra#1#2{\ket{#1}\!\bra{#2}}
\global\long\def\real{\operatorname{Re}}
\global\long\def\one{\mathds{1}}

\newcommand{\kommentar}[1]{}

\newcommand{\tracenorm}[1]{\| #1 \|_1}

\NewDocumentCommand\opti{smmm>{\SplitList{;}}m} {
\begingroup%
\setlength{\belowdisplayskip}{-0.6\baselineskip}%
\IfBooleanTF{#1}{%
    \begin{alignat*}{2}
        & \underset{#3}{\text{#2}} & & #4 \\
        & \text{subject to~~}
        \ProcessList{#5}{ \insertopticonst }
        & &
    \end{alignat*}%
    }{%
    \begin{alignat}{2}
        & \underset{#3}{\text{#2}} & & #4 \\
        & \text{subject to~~}
        \ProcessList{#5}{ \insertopticonst }
        & & \nonumber
    \end{alignat}%
    }%
\endgroup%
}%
\newcommand\insertopticonst[1]{& & #1\\&}

\begin{document}

\title{Bound entangled Bell diagonal states of unequal local dimensions, and their witnesses}

\author{Johannes Moerland}

\affiliation{Institut für Theoretische Physik III, Heinrich-Heine-Universität Düsseldorf, Universitätsstr.~1, D-40225 Düsseldorf, Germany}

\affiliation{Universität Göttingen, Friedrich-Hund-Platz 1, D-37073 Göttingen, Germany}

\author{Nikolai Wyderka}

\affiliation{Institut für Theoretische Physik III, Heinrich-Heine-Universität Düsseldorf, Universitätsstr.~1, D-40225 Düsseldorf, Germany}

\author{Hermann Kampermann}

\affiliation{Institut für Theoretische Physik III, Heinrich-Heine-Universität Düsseldorf, Universitätsstr.~1, D-40225 Düsseldorf, Germany}

\author{Dagmar Bruß}

\affiliation{Institut für Theoretische Physik III, Heinrich-Heine-Universität Düsseldorf, Universitätsstr.~1, D-40225 Düsseldorf, Germany}

\date{\today}

\begin{abstract}
Bell diagonal states constitute a well-studied family of bipartite quantum states that arise naturally in various contexts in quantum information. In this paper we generalize the notion of Bell diagonal states to the case of unequal local dimensions and investigate their entanglement properties. We extend the family of entanglement criteria of Sarbicki et al.~to non-Hermitian operator bases to construct entanglement witnesses for the class of generalized Bell diagonal states. We then show how to optimize the witnesses with respect to noise robustness. Finally, we construct bound entangled states that are detected by these witnesses, but not by the usual computable cross norm or realignment and de~Vicente criteria.
\end{abstract}
\maketitle

\section{Introduction}

In bipartite quantum systems, Bell diagonal states (BDS) constitute an important class of states which occur naturally in studies of nonlocality and entanglement theory \cite{batle2011nonlocality, riedel2021bell}. They are characterized by the fact that they are diagonal in the Bell basis, i.e., the basis obtained by certain local unitary rotations of maximally entangled bipartite states. The class contains both the maximally mixed state, as well as maximally entangled states. Moreover, it contains bound entangled states, which are states that exhibit entanglement, but cannot be distilled to pure Bell pairs \cite{horodecki1998mixed}. As such, the class of these states is a popular playground to gain insight in the structure of states and entanglement in high dimensional systems. What is more, these states are relevant in the context of quantum key distribution, as in certain protocols the shared states between the two parties can be assumed to be Bell diagonal, thereby simplifying the analysis of security \cite{grasselli2021quantum}.

There exist several generalizations of the class of Bell-diagonal states, most of which extend the usual Bell-diagonal states to the case of multi-qubit systems \cite{yu2005evolution,batle2010nonlocality,rafsanjani2012genuinely,batle2016nonlocality}.
In \cite{uchida2015entangled}, the authors extend the notion to the multi-qu$d$it scenario.
All of these generalizations have in common that they range only over composite systems of equal local dimension.

However, there exist several physical systems where this assumption is not met, e.g., for photons interacting with ions \cite{bock2018high}, or systems using photonic polarization and time of arrival degrees of freedom \cite{pilnyak2020encoding}. 
This indicates that it is useful to extend the definition of Bell-diagonal states to the case where the local Hilbert space dimensions do not coincide. In this paper, we propose such an extension in the bipartite case and analyze their entanglement features.

Entanglement is an important resource for imminent applications of quantum technology, as its presence guarantees genuine quantum correlations between several systems. However, its efficient and reliable detection remains a hard challenge, particularly for quantum systems of large dimensions, as the decision problem of whether a given state is entangled or not, is known to be NP hard \cite{gurvits2003classical, gharibian2008strong}. This remains true for the subclass of Bell-diagonal states, despite recent progress \cite{popp2022almost}.

Nevertheless, many different criteria are known which are sufficient for the presence of entanglement. An easily applicable criterion for entanglement arises from so-called
\emph{entanglement witnesses}, which are observables on the composite Hilbert space that yield non-negative
expectation values for all separable states, and strictly negative expectation values for some entangled
states \cite{guhne2009entanglement}. The existence of such linear operators is guaranteed by the Hahn-Banach theorem as the set of
separable density operators is a convex subset of all density operators.
Naturally, the following questions arise: \emph{Given an entangled density operator, how do we construct
an entanglement witness that is capable of detecting its entanglement? And out of these witnesses, which one
is suited best?} Generally, these two questions are difficult to answer. However, in this paper, we restrict
our attention to entangled states that can be detected using a powerful criterion derived by 
Sarbicki, Scala and Chru{\'s}ci{\'n}ski \cite{class-of-ews}. Their criterion allows for the construction of entanglement witnesses, 
 but assumes Hermitian operator bases. We will refer to said criterion and witnesses as \emph{SSC bound} and \emph{SSC witnesses}, respectively. 
As the generalized Bell-diagonal states are naturally expressed in terms of a non-Hermitian operator basis, we generalize the SSC construction to allow for arbitrary orthonormal operator bases of the local Hilbert spaces.
Finally, we show how to
optimize over the class of arising SSC witnesses.
While our results on the witnesses are universally applicable, we focus on their use in the case of generalized Bell diagonal states.

This paper is organized as follows. In Section~\ref{sec:bdsunequal}, we introduce the family of generalized BDS. As these states are naturally defined in a non-Hermitian operator basis, we generalize in Section~\ref{sec:sscwitnessfornon} the entanglement criterion of Ref.~\cite{class-of-ews} to the case of non-Hermitian bases. We then reduce the optimization process required to derive entanglement witnesses to a singular value decomposition and isolate optimal choices of parameters of the criterion by relating it to the white noise robustness of entanglement detection. In Section~\ref{sec:sscwitnessbds} we explicitly investigate the criterion for the introduced class of generalized BDS and show that there exist detectable Bell diagonal states of unequal dimension which are bound entangled, but are neither detected by the usual computable cross norm or realignment (CCNR) criterion \cite{chen2002matrix, rudolph2005further}, nor the de~Vicente criterion \cite{vicente2007separability, de2008further}.

\section{Bell diagonal states of unequal local dimensions}\label{sec:bdsunequal}
Let $\mathcal{H} = \mathcal{H}_A \otimes \mathcal{H}_B$ be a
bipartite Hilbert space with local dimensions $d_A$ and $d_B$,
respectively. Without loss of generality, we assume
$d_A \le d_B$.
Throughout this paper, we denote by $S\in\{A,B\}$ the subsystem of interest.\par
We would like to construct generalized Bell mixtures acting on a
bipartite Hilbert space where generally, $d_A \neq d_B$. To do so, we fix an orthonormal
basis $\{\ket{i}_S\}_{i=0}^{d_S-1}$ on each local space $\mathcal{H}_S$.
This allows us to define the unitary shift operator 
\begin{equation}
X_S : \mathcal{H}_S \to \mathcal{H}_S, \ket{i}_S \mapsto \ket{i\oplus 1}_S,
\end{equation}
where $\oplus$ denotes addition modulo $d_S$, as
well as the unitary clock operator 
\begin{equation}
Z_S : \mathcal{H}_S \to \mathcal{H}_S, \ket{i}_S \mapsto \omega_S^{i}\ket{i}_S,
\end{equation}
where
$\omega_S = \exp(2\pi \mathrm{i} / d_S)$ is the corresponding
root of unity for subsystem $S$. Both of the above definitions are extended to all of $\mathcal{H}_S$ via linearity. Note that for each subsystem $S$,
the operators $\{X_S^\alpha Z_S^\beta\}_{\alpha, \beta = 0}^{d_S - 1}$
form the so-called \emph{Heisenberg-Weyl basis} of linear operators acting on $S$ \cite{zyczkowski2006geometry}. In particular, they obey
\begin{equation}\label{eq:hw-orthogonality}
    \trace{\left(\left(X_S^\alpha Z_S^\beta\right)^\dagger X_S^\mu Z_S^\nu \right)}
    = d_S\delta_{\alpha\mu}\delta_{\beta\nu},
\end{equation}
i.e., they are pairwise orthogonal w.r.t.~the Hilbert-Schmidt inner product. Note that this also implies that for $(\alpha,\beta)\neq (0,0)$, the operator $X_S^\alpha Z_S^\beta$ is traceless. Here and in the following, the exponents of the clock and shift operators may be interpreted modulo $d_S$ as $X_S^{d_S} = Z_S^{d_S} = \one_S$. \par
We now construct a class of quantum states that exhibit particularly nice symmetry properties. To do so, we define the generalized Bell state
\begin{equation}\label{eq:bell00}
    \ket{\phi^{00}} := \frac{1}{\sqrt{d_A}}\sum_{i=0}^{d_A - 1}
    \ket{i}_A\otimes \ket{i}_B,
\end{equation}
and extend it to a basis of the bipartite Hilbert space via
\begin{equation}
    \ket{\phi^{\alpha\beta}} :=
    \left(Z^\alpha_A \otimes X^\beta_B \right) \ket{\phi^{00}}.
\end{equation}
In the following, we will omit the index denoting the subsystem, as the position of the operators and vectors leaves no ambiguity. One can verify that the $\ket{\phi^{\alpha\beta}}$ are all pairwise orthogonal. 
For a given probability distribution $\{p_{\alpha\beta}\}$, i.e., $p_{\alpha\beta}\geq 0$ and $\sum_{\alpha,\beta} p_{\alpha\beta} = 1$, we denote by $P$ the $d_A\times d_B$ matrix with entries $p_{\alpha\beta}$, i.e.,
\begin{align}
P = \begin{pmatrix}p_{00} & \ldots & p_{0,d_B-1} \\
\vdots & & \vdots\\
p_{d_A-1,0} & \ldots & p_{d_A-1,d_B-1}
\end{pmatrix},
\end{align}
and define
\begin{equation}\label{eq:GBM}
    \rho_P := \sum_{\alpha=0}^{d_A-1}\sum_{\beta=0}^{d_B-1}
    p_{\alpha\beta}
    \ketbra{\phi^{\alpha\beta}}{\phi^{\alpha\beta}}.
\end{equation}
These density operators are diagonal in the generalized Bell basis $\ket{\phi^{\alpha\beta}}$, hence, we will call them \emph{Bell diagonal states}, or \emph{BDS} for short.

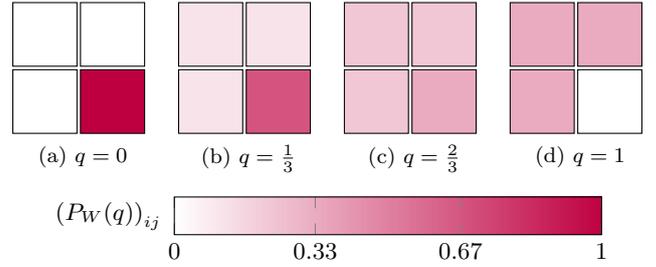
\begin{figure}
    \centering
    \begin{subfloat}[][$q = 0$]
         \wernerzero
    \end{subfloat}
    \begin{subfloat}[][$q = \frac13$]
         \wernerone
    \end{subfloat}
    \begin{subfloat}[][$q = \frac23$]
         \wernertwo
    \end{subfloat}
    \begin{subfloat}[][$q = 1$]
         \wernerthree
    \end{subfloat}

\begin{tikzpicture}
\begin{axis}[
    hide axis,
    scale only axis,
    height=0pt,
    width=0pt,
    colormap={slategraywhite}{rgb255=(255,255,255) rgb255=(191,0,64)},
    colorbar horizontal,
    point meta min=0,
    point meta max=1,
    colorbar style={
        width=0.65\columnwidth,
        xtick={0,0.33,0.67,1},
        ylabel=\rotatebox{-90}{$\left(P_W(q)\right)_{ij}\,$}
    }]
    \addplot [draw=none] coordinates {(0,0)};
\end{axis}
\end{tikzpicture}
    
    \caption{Graphical representations of the entries of the probability matrices $P_W(q)$ of the Bell diagonal Werner states from Eqs.~\eqref{eq:werner} and \eqref{eq:werner-matrix} for (a) $\rho_W(0)$, (b) $\rho_W(\frac13)$, (c) $\rho_W(\frac23)$ and (d) $\rho_W(1)$.
    }
    \label{fig:werner}
\end{figure}

A simple example for Bell diagonal states in $2\times 2$ dimensions is the family of Werner states \cite{werner1989quantum}, given by
\begin{align} 
    \rho_W(q) & = \frac q3\big(\ketbra{\phi^{00}}{\phi^{00}} + \ketbra{\phi^{01}}{\phi^{01}} + \ketbra{\phi^{10}}{\phi^{10}}\big) \nonumber\\
              &\pushright{+ (1-q)\ketbra{\phi^{11}}{\phi^{11}}},\label{eq:werner}
\end{align}
where $q\in[0,1]$. The probability matrix of such a state $\rho_W (q)$ thus reads
\begin{equation}
    P_W(q) = 
    \begin{pmatrix}
        \frac{q}{3} & \frac{q}{3}\\
        \frac{q}{3} & 1-q
    \end{pmatrix}.\label{eq:werner-matrix}
\end{equation}
In Fig.~\ref{fig:werner}, we show a graphical grid representation of the probability matrices of the Werner states $\rho_W(0)$, $\rho_W(\frac13)$, $\rho_W(\frac23)$ and $\rho_W(1)$.

Using Eq.~\eqref{eq:hw-orthogonality}, a BDS can be decomposed in the Heisenberg-Weyl basis as
\begin{equation}\label{eq:bds-hw-uneq-full}
    \rho_P = \frac{1}{d_A d_B}
    \sum_{\kappa,\mu=0}^{d_A-1}
    \sum_{\nu=0}^{d_B-1}s_{\kappa\nu} 
    \lambda_{\mu\nu} X^\mu Z^\kappa
    \otimes X^\mu Z^{-\nu},
\end{equation}
where the $s_{\kappa\nu}$ are given by
\begin{equation}
    s_{\kappa\nu} =
    \frac{1}{d_A} \sum_{j=0}^{d_A-1} \left(\omega_A^\kappa \omega_B^{-\nu}\right)^j
\end{equation}
and the $\lambda_{\mu\nu}$ are the Fourier transformed probabilities
\begin{equation}
    \lambda_{\mu\nu} =
    \sum_{\alpha=0}^{d_A-1}\sum_{\beta=0}^{d_B-1}
    p_{\alpha\beta} 
    \omega_A^{\alpha\cdot \mu}\omega_B^{\beta\cdot \nu}.
\end{equation}
We denote by $\Lambda$ the matrix with entries $\lambda_{\mu\nu}$. We will later connect the latter to properties of the state. 

Conversely, given the coefficients $\lambda_{\mu\nu}$, the probabilities $p_{\alpha\beta}$ can be obtained by means of the inverse Fourier transform, i.e.,
\begin{equation}
    p_{\alpha\beta} = \frac{1}{d_A d_B}
    \sum_{\mu=0}^{d_A-1}\sum_{\nu=0}^{d_B-1}
    \lambda_{\mu\nu}
    \omega_A^{-\alpha\cdot \mu} \omega_B^{-\beta\cdot \nu}.
\end{equation}
As an example, consider the state $\rho_P = \ketbra{\phi^{00}}{\phi^{00}}$, i.e., $p_{00}=1$. Its Fourier coefficients are given by $\lambda_{\mu\nu} = 1$ for all $\mu$ and $\nu$.
As it turns out, this Fourier picture yields insight into the entanglement properties of the BDS.

Whether a given Fourier matrix $\Lambda$ yields a proper probability distribution, i.e., $0\leq p_{\alpha\beta}\leq 1$ and $\sum_{\alpha\beta} p_{\alpha\beta} = 1$, is cumbersome to check without actually performing the transformation. Apart from the trivial conditions $\lambda_{00} = 1$ from normalization and $\lambda_{\mu\nu} = \lambda^*_{-\mu,-\nu}$ due to reality of the probabilities (note that indices are taken modulo $d_A$ and $d_B$, respectively, and $z^*$ denotes the complex conjugate of $z$), a set of necessary conditions is given by the Herglotz-Bochner theorem  \cite{rudin2017fourier, maruyama2017herglotz} that yields conditions on the Fourier coefficients for the probabilities to be positive. In particular, it implies that the $r \times r$-dimensional Toeplitz matrices $\mathcal{T}_{(f_A,f_B)}^{(r)}$ with entries
\begin{align}
    [\mathcal{T}_{(f_A,f_B)}^{(r)} ]_{\mu\nu} = \lambda_{(\mu-\nu)f_A,(\mu-\nu)f_B}
\end{align}
need to be positive semidefinite for all $0\leq f_A \leq d_A-1, 0\leq f_B \leq d_B-1$ and $r \geq 1$.  As an example, the Toeplitz matrices of size $2\times 2$ yield the conditions $\vert \lambda_{\mu\nu}\vert^2 \leq \vert \lambda_{00}\vert^2 = 1$ for all $\mu$ and $\nu$.\par

If we define a modified clock operator via its action on the computational basis vectors
\begin{equation}
    \Tilde{Z}: \mathcal{H}_A\to \mathcal{H}_A, \ket{i}_A\mapsto \omega_B^i \ket{i}_A
\end{equation}
(note that we multiply vectors from $\mathcal H _A$ by powers of $\omega_B$), as well as the operators
\begin{equation}
    T(\mu,\nu)=
    X^\mu \Tilde{Z}^\nu\otimes X^\mu Z^{-\nu}
\end{equation}
acting on the bipartite Hilbert space,
the BDS given by \eqref{eq:bds-hw-uneq-full} becomes diagonal:
\begin{equation}\label{eq:BD-HW-uneq}
    \rho_P = \frac{1}{d_A d_B}
    \sum_{\mu=0}^{d_A-1}\sum_{\nu=0}^{d_B-1}
    \lambda_{\mu\nu} T(\mu,\nu).
\end{equation}
Note however that for $d_A \neq d_B$, this is not a proper operator Schmidt decomposition of the state because $\{X^\mu \Tilde{Z}^\nu\}$ is an over-complete generating set of linear operators acting on $\mathcal{H}_A$.\par
For the case of equal dimensions $d_A = d_B =: d$, the operators $\Tilde{Z}$ and $Z$ are identical, and therefore, the state has the particularly nice form
\begin{equation}\label{eq:bds_bloch}
    \rho_P = \frac{1}{d^2}
    \sum_{\mu,\nu}\lambda_{\mu\nu}
    X^\mu Z^\nu \otimes X^\mu Z^{-\nu}.
\end{equation}
Thus, one obtains the
usual definition of Bell diagonal states (up to relabeling), as
found, e.g., in \cite{BDStatesSymmetry}.\par
For $d_A = d_B$, it is known that there exist non-entangling quantum channels that reduce arbitrary bipartite states to Bell diagonal ones \cite{popp2023special}. If one can detect entanglement in the resulting BDS, then the original state must already have been entangled. We will now show that a similar result also holds if $d_A < d_B$.
To that end, observe that pairwise tensor products of the Heisen\-berg\--Weyl operators together with the $T(\mu,\nu)$ and $\omega_A^\alpha \omega_B^\beta \one\otimes \one$ generate a finite group $G$ under composition (multiplication).
We denote by $\Ad_g (h) := ghg^{-1},\, g,h\in G$ the adjoint inner group automorphism, and extend it to linear combinations of group elements via linearity in $h$.
For every $q \in [0,1]$, we define a map on $\operatorname{End}(\mathcal{H})$ (i.e., the space of linear operators acting on $\mathcal{H}=\mathcal{H}_A\otimes\mathcal{H}_B$):
\begin{equation}\begin{gathered}\label{eq:channel-hw}
\Phi_q : \operatorname{End}(\mathcal{H})\to \operatorname{End}(\mathcal{H}),\\ \rho \mapsto (1-q)\rho + \frac{q}{d_A d_B}
    \sum_{\mu=0}^{d_A-1} \sum_{\nu=0}^{d_B-1}
    \Ad_{T(\mu,\nu)}(\rho).
\end{gathered}
\end{equation}
Note that these maps are trace preserving and positive (i.e., they map positive operators to positive operators). They can be used to convert quantum states into their Bell diagonal contribution:
\begin{proposition}
    For all linear operators $\rho$ on the bipartite Hilbert space and for all $q \in [0,1]$, it holds that
\begin{align}\label{eq:channel-bell}
    \Phi_q(\rho)
    =&(1-q)\rho + q
    \sum_{\alpha=0}^{d_A-1}\sum_{\beta=0}^{d_B-1}
    \braket{\phi^{\alpha\beta}|\rho|
    \phi^{\alpha\beta}}
    \ket{\phi^{\alpha\beta}}\bra{\phi^{\alpha\beta}}.
\end{align}
Restricting the domain of the maps $\Phi_q$ to density operators gives rise to a family of non-entangling quantum channels. For $q = 1$, this channel projects each state onto the set of Bell diagonal states.
\end{proposition}
\begin{proof}
    To verify the equality, expand $\rho$ in the Heisenberg-Weyl basis and use the identities 
    \begin{align}
        Z_S^\nu X_S^\mu = \omega_S^{\mu\cdot \nu} X_S^\mu Z_S^\nu, &&
        \Tilde{Z}^\nu X_A^\mu = \omega_B^{\mu\cdot \nu} X_A^\mu \Tilde{Z}^\nu.
    \end{align}
    Positivity and trace preservation follow directly from the definition of the $\Phi_q$, hence, by restricting the domain to density operators, we can view them as quantum channels.
    In doing so, one sees that the $\Phi_q$ are non-entangling because they are convex combinations of local unitaries (see \eqref{eq:channel-hw}). By \eqref{eq:channel-bell}, it is clear that for $q = 1$, only the Bell diagonal part remains. 
\end{proof}
Hence, if we use this channel with $q=1$ for a quantum state, and detect entanglement in the resulting state, we know that the original state must have been entangled. Therefore, the channel can be used to characterize entanglement of arbitrary states.\par
To see the dephasing channel in action, consider the $2\times 3$ maximally entangled $|\phi^{00}\rangle$ and construct the one-parameter family of states
\begin{equation}\label{eq:phi-theta}
    |\phi_\theta\rangle = Z^\theta \otimes \one |\phi^{00}\rangle = |00\rangle + e^{\textrm{i}\theta\pi}|11\rangle ,
\end{equation}
where $0\le\theta \le 1$.
Applying the channel $\Phi_1$ yields
\begin{align}
    \Phi_1(|\phi_\theta\rangle\langle\phi_\theta|)
    = \cos ^2 (\theta\pi / 2)&|\phi^{00}\rangle\langle\phi^{00}|\\+ 
    \sin ^2 (\theta\pi / 2)&|\phi^{10}\rangle\langle\phi^{10}|,
\end{align}
which is also entangled whenever $\theta\neq 1/2$, as can be verified with the CCNR or the PPT criterion. It follows that for $\theta\neq 1/2$, the $\ket{\psi_\theta}$ must have been entangled.\par
For $d_A = d_B = d$, the Fourier picture allows to obtain a simple sufficient entanglement test based on the CCNR criterion \cite{chen2002matrix, rudolph2005further}: As the operators $\{X^\mu Z^\nu\}$ form an orthogonal operator basis, Eq.~(\ref{eq:bds_bloch}) constitutes a proper operator Schmidt decomposition, after one absorbs the phases of the $\lambda_{\mu\nu}$ into one of the bases. This means that $\rho_P$ is entangled if its CCNR value exceeds $d$ \cite{chen2002matrix, rudolph2005further}, i.e.,
\begin{align}\label{eq:ccnr}
    \operatorname{CCNR}(\rho_P):=\sum_{\mu,\nu=0}^{d-1} \vert \lambda_{\mu\nu} \vert > d.
\end{align}

\section{Constructing bound entangled Bell diagonal states}

A particularly interesting class of quantum states is that of bound entangled ones. These states cannot be used to distill highly entangled states by means of local operations and classical communication \cite{horodecki1998mixed}. Hence, these states are usually weakly entangled and their entanglement is notoriously difficult to detect. To benchmark the employed entanglement criteria,  we will construct examples of Bell diagonal bound entangled states. To that end, we make use of the fact that entangled states that are not detected by the PPT criterion are bound entangled \cite{horodecki1998mixed}. Remember that the PPT criterion states that for separable states $\sigma$, $\sigma^{T_A} := (T \otimes \one)(\sigma) \geq 0$, where $T$ denotes the transposition map in the computational basis \cite{peres1996separability, horodecki1997separability}. 

In terms of Eq.~(\ref{eq:bds_bloch}), the partial transposition maps the operator on system $A$ from $X^\mu Z^{\nu}$ to $\omega^{-\mu\cdot \nu} X^{-\mu}Z^{\nu}$. The positivity of the resulting state is generally hard to check.

We start our constructions by considering the case where $d_A=d_B=d$, and return to the case of unequal local dimensions afterwards. To ensure that the constructed states are indeed entangled, we make use of the CCNR criterion discussed in \eqref{eq:ccnr}.
Constructing $\rho_P$ such that $\rho_P^{T_A} \geq 0$ is not trivial, which is why one usually uses tricks to ensure it \cite{bruss2000construction, smolin2001four, lockhart2018entanglement}. Here, we use two different methods to construct states with an inherently positive partial transpose. First, we demand that the partial transpose is a multiple of a projector, i.e., $(\rho_P^{T_A})^2 \propto \rho_P^{T_A}$, implying positivity. Second, we construct $\rho_P$ such that $\rho_P = \rho_P^{T_A}$, similarly to the approach in Ref.~\cite{bruss2000construction}.

\subsection{States with partial transpose being proportional to a projector}

Our first approach demands that $(\rho_P^{T_A})^2 \propto \rho_P^{T_A}$. To simplify things, we restrict ourselves to a specific subfamily of Bell diagonal states, which we call dichotomous Bell diagonal states. These are states where each Bell state is either non-present or comes with the same, fixed probability. 
Such states can be characterized uniquely by the set $\mathcal S=\{(\alpha,\beta)\,:\,p_{\alpha\beta} \neq 0\}$. 
In particular, it holds for the probabilities
\begin{align}
    p_{\alpha\beta} \in \left\{ 0,\frac1{\vert \mathcal S \vert} \right\}.
\end{align}

Writing these states as 
\begin{align}\label{eq:dbds}
    \rho_\mathcal S := \frac1{\vert \mathcal S\vert}\sum_{(\alpha,\beta) \in \mathcal S}\ketbra{\phi^{\alpha\beta}}{\phi^{\alpha\beta}}
\end{align}
and demanding $(\rho_\mathcal S^{T_A})^2 \propto \rho_\mathcal S^{T_A}$ yields after a lengthy but straightforward calculation 
\begin{align}\label{eq:projector}
    \sum_{\substack{(\alpha,\beta) \in \mathcal S\\(\mu,\nu) \in \bar{\mathcal S}}} \omega^{(\alpha,\beta) \wedge (\mu,\nu)} X^{\mu - \alpha} Z^{\nu - \beta} \otimes (X^{\mu - \alpha} Z^{\nu - \beta})^\dagger = 0,
\end{align}
where $\bar{\mathcal S}$ denotes the complement of $\mathcal S$ and $(\alpha,\beta) \wedge (\mu,\nu) = \alpha\nu - \beta \mu$ denotes the standard symplectic form. Note that this equation yields one condition for each fixed displacement $\Delta = (\mu,\nu) - (\alpha,\beta)$, with subtraction (and addition) defined component wise. Graphically, this condition can be read as follows. Draw a grid of size $d\times d$ and color the cells corresponding to points in $\mathcal S$. For each displacement $\Delta$, check all colored cells $(\alpha,\beta)$ on the grid, such that the displacement links a colored to a non-colored cell $(\mu,\nu) = (\alpha,\beta)+\Delta$ (note that all operations are again taken modulo $d$). For each such colored cell $(\alpha, \beta)$, calculate the phase $\omega^{(\alpha,\beta)\wedge \Delta}$. The condition in Eq.~(\ref{eq:projector}) now states that for all displacements $\Delta$, the sum of the phases of the admissible cells must vanish.

We visualize this procedure by the following example in $d=3$. We choose $\mathcal S=\{(0,0), (1,1), (1,2), (2,1), (2,2)\}$, which would correspond to the Bell-diagonal state $\rho_P$ given by probabilities
\begin{align}
    P = \frac15\begin{pmatrix}1 & 0 & 0 \\ 0 & 1 & 1 \\ 0 & 1 & 1\end{pmatrix}.
\end{align}
This can be equivalently visualized by the grid displayed in Fig.~\ref{fig:pattern3A}.

\begin{figure}
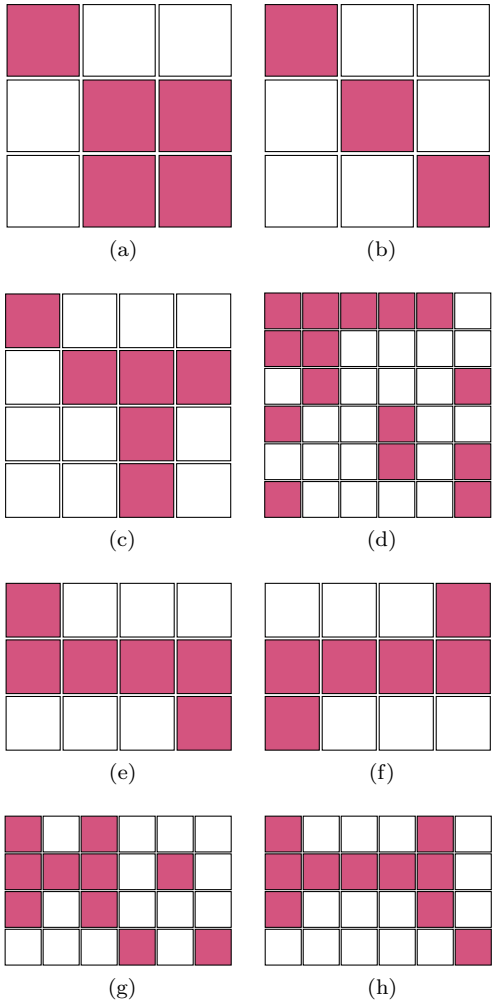

    \centering
    \begin{subfloat}[\label{fig:pattern3A}]
         \patternthree
     \end{subfloat}
     \begin{subfloat}[\label{fig:pattern3B}]
         \patternthreetwo
     \end{subfloat}
     \begin{subfloat}[\label{fig:pattern4}]
         \patternfour
     \end{subfloat}
     \begin{subfloat}[\label{fig:pattern6}]
         \patternsix
     \end{subfloat}
     \begin{subfloat}[\label{fig:pattern34A}]
         \patternthreefour
     \end{subfloat}
     \begin{subfloat}[\label{fig:pattern34B}]
         \patternthreefourtwo
     \end{subfloat}
     \begin{subfloat}[\label{fig:pattern46A}]
         \patternfoursix 
     \end{subfloat}  
     \begin{subfloat}[\label{fig:pattern46B}]
         \patternfoursixtwo
     \end{subfloat}
    \caption{Grid representation of some of the dichotomous BDS considered in the text. Displayed are (a) a $3\times 3$ dimensional BDS that does not fulfill the phase condition from Eq.~(\ref{eq:projector}) for guaranteeing PPT; (b) a state which does fulfill the phase condition, thus is PPT but not entangled; (c) the unique dichotomous $4\times 4$ dimensional Bell diagonal state which fulfills the phase condition and violates the CCNR criterion maximally; (d) a $6\times 6$ dimensional BDS which is $6$-displacement homogeneous (see Def.~\ref{def:homogeneity}) but not PPT; (e) and (f) the two bound entangled dichotomous $4\times 6$ dimensional BDS.}
    \label{fig:patterns}
\end{figure}

Now, consider the displacement $\Delta=(1,1)$. It links  $\mathcal S$ to $\bar{\mathcal S}$ only at places $(1,2)$ and $(2,1)$ with corresponding phases $\omega^{(1,2)\wedge (1,1)} = \omega^{-1}$ and $\omega^{(2,1)\wedge (1,1)}=\omega$. Thus, the corresponding state does not have the projector property. In contrast, consider the state defined by choosing $\mathcal S=\{(0,0), (1,1), (2,2)\}$, displayed in Fig.~\ref{fig:pattern3B}.
Here, choosing for example $\Delta = (0,1)$, links all three points in $\mathcal S$ to $\bar{\mathcal S}$, with sum of phases 
$\omega^{(0,0)\wedge(0,1)} + \omega^{(1,1)\wedge(0,1)} + \omega^{(2,2)\wedge(0,1)} = 1 + \omega + \omega^2 = 0$. The same applies to all other $\Delta$ apart from $\Delta = (1,1)$ and $\Delta = (2,2)$, which link no point at all from $\mathcal S$ to $\bar{\mathcal S}$, and are therefore trivially fulfilled. 

While this last example does have a projector-like partial transpose, it is not detected as entangled by the CCNR criterion: it reaches a CCNR value of $\operatorname{CCNR}(\rho_\mathcal S) = 3$. In fact, by exhaustively checking all dichotomous states with $d=3$, we can say that there are no dichotomous Bell diagonal states with entanglement detected by CCNR. However, in $d=4$, we can find (up to symmetries) exactly one solution that is detected. It is displayed in Fig.~\ref{fig:pattern4} with its probability matrix given by
\begin{align}\label{eq:peacock}
    P = \frac16\begin{pmatrix} 1 & 0 & 0 & 0 \\ 0 & 1 & 1 & 1 \\ 0 & 0 & 1 & 0 \\ 0 & 0 & 1 & 0\end{pmatrix}.
\end{align}
This state is very similar to the bound entangled state constructed in Ref.~\cite{benatti2004non}.\footnote{While we could not find local unitary transformations to transform the states into another, we could not find a local unitary invariant which differs.} It has a CCNR value of $6 > d=4$. Numerical search indicates that this is the largest violation any (even non-dichotomous) Bell diagonal state with positive partial transpose exposes in this dimension. 
Furthermore, it has additional nice properties: For each displacement $\Delta \neq (0,0)$, exactly four points in $\mathcal S$ are linked to points in $\bar{\mathcal S}$. This motivates us to search for such homogeneous solutions, which we capture in the following definition:

\begin{definition}\label{def:homogeneity}
A dichotomous Bell diagonal state $\rho_\mathcal S$ as defined in Eq.~(\ref{eq:dbds}) is called \emph{$k$-displacement homogeneous} if each displacement $\Delta \neq (0,0)$ links exactly $k$ of the points in $\mathcal S$ to points in $\bar{\mathcal S}$. 
\end{definition}

However, such states can only be found for certain parameters: Let us denote the dimension by $d$, the rank of the state by $\vert \mathcal S\vert$ and the displacement homogeneity by $k$.
As there are $d^2-1$ different non-vanishing displacement tuples, the total number of links between $\mathcal S$ and $\bar{\mathcal S}$ is given by $(d^2-1)k$.

On the other hand, there are $\vert \mathcal S \vert(d^2-\vert \mathcal S \vert)$ possible links between $\mathcal S$ and $\bar{\mathcal S}$, each contributing to exactly one of the shifts.
Setting both equal, we obtain the diophantine equation \cite{diophantus300arithmetica}
\begin{align}\label{eq:dio}
    (k-\vert \mathcal S \vert) d^2-k+\vert \mathcal S\vert^2 = 0,
\end{align}
and all we have to do is to find its integer solutions, which is, in general, a hard task. Checking numerically for solutions, with fixed $d$ we obtain
the solutions given in Table~\ref{tab:diosol}. A few comments are in order:

First, note that we do not list the following trivial solutions: For each solution $\mathcal S$, also $\bar{\mathcal S}$ is a valid solution. We list only the solution with the smaller number of elements in $\mathcal S$, as they exhibit a larger CCNR violation. Setting $k=0$, we can always find the maximally mixed state as a solution, where $\mathcal S$ is the full set. Setting $k=1$, we always find the pure Bell states as solutions with $\vert \mathcal S \vert = 1$. These can never exhibit a positive partial transpose and are therefore omitted as well.

\begin{table}[t]
    \centering
    \begin{tabular}{r|r|r|r|l}
        $d$ & $\vert \mathcal S \vert$ & $k$ & CCNR$-d$ & entangled dichotomous BDS exists?\\
        \hline
4& 6&4& 2 & yes\\ 
5& 9&6& 2.53197 & no\\ 
6& 15&9& 2 & yes, but not PPT\\ 
7& 16&11& 3.94987\\ 
8& 28&16& 2\\ 
9& 16&13& 10.0278\\ 
10& 45&25& 2\\ 
11& 16&14& 18.0624\\ 
11& 25&20& 11.4663\\ 
11& 40&27& 5.58846\\ 
12& 66&36& 2
    \end{tabular}
    \caption{All non-trivial integer solutions to the diophantine Eq.~(\ref{eq:dio}) for $d\leq 12$, where $d$ denotes the local dimension, $\vert \mathcal S\vert$ denotes the rank and $k$ the displacement homogeneity, i.e., the number of links from $\mathcal S$ to $\bar{\mathcal S}$ for each displacement $\Delta$. The CCNR violation is calculated via Lemma~\ref{lem:ccnr}.  While these solutions in principle allow for the existence of dichotomous states similar to the one in Eq.~(\ref{eq:peacock}), we did not find any bound entangled solution for $d>4$, and could exclude their existence for $d=5$ and $d=6$.}
    \label{tab:diosol}
\end{table}

The remaining solution pairs are displayed in the table.

Second, we can construct a family of solutions for all even $d$ by setting $\vert \mathcal S\vert = \binom{d}{2}$ and $k=\frac{d^2}{4}$. These are those solutions yielding an excess CCNR value of $2$.

Note that we also list the CCNR value, which is fixed for this kind of states:
\begin{lemma}\label{lem:ccnr}
A $k$-displacement homogeneous dichotomous Bell diagonal state $\rho_\mathcal S$ defined by the set $\mathcal S \subset \{(\alpha,\beta)\,:\,0\leq \alpha,\beta < d\}$ as given in Eq.~(\ref{eq:dbds})  has a CCNR value of $\operatorname{CCNR}(\rho_\mathcal S) = 1+ (d^2-1)\frac{\sqrt{k}}{\vert \mathcal S\vert}$.
\end{lemma}
\begin{proof}
Consider the Bloch component
\begin{align}
    \lambda_{\mu\nu} &= \sum_{\alpha\beta}p_{\alpha\beta} \omega^{\alpha\cdot\mu + \beta\cdot \nu} \\
    &= \frac{1}{\vert \mathcal S\vert}\sum_{(\alpha,\beta)\in \mathcal S} \omega^{(\alpha,\beta)\cdot(\mu,\nu)},
\end{align}
where $(\alpha,\beta)\cdot(\mu,\nu) := \alpha\cdot\mu + \beta\cdot\nu$. The squared magnitude of $\lambda_{\mu\nu}$ is given by
\begin{align}\label{eq:lambdasquare}
    \vert \lambda_{\mu\nu} \vert^2 &= \frac{1}{\vert \mathcal S\vert^2}\sum_{(\alpha,\beta)\in \mathcal S}\sum_{(\alpha^\prime,\beta^\prime)\in \mathcal S} \omega^{[(\alpha,\beta)-(\alpha^\prime,\beta^\prime)]\cdot(\mu,\nu)}.
\end{align}
Now, the $k$-displacement homogeneity of the links from $\mathcal S$ to $\bar{\mathcal S}$ implies also a displacement homogeneity of the links from $\mathcal S$ to itself: as the the total number of links for fixed $\Delta$ from $\mathcal S$ to anything is given by $\vert \mathcal S\vert$, and $k$ of them link to $\bar{\mathcal S}$, the other $\vert \mathcal S\vert - k$ must link to $\mathcal S$ again. The single exception is $\Delta=(0,0)$, as this shift always leads to $\mathcal S$ again. Thus, we can replace the double sum in Eq.~(\ref{eq:lambdasquare}) by a single sum and a factor:
\begin{align}
    \vert \lambda_{\mu\nu} \vert^2 &= \frac1{\vert \mathcal S\vert^2} \left[\vert \mathcal S \vert + (\vert \mathcal S \vert - k)\sum_{\Delta\neq(0,0)} \omega^{\Delta \cdot (\mu,\nu)}\right]\\
    &= \frac1{\vert \mathcal S\vert^2} \left[\vert \mathcal S \vert + (\vert \mathcal S \vert - k)(d^2\delta_{\mu,0}\delta_{\nu,0}-1) \right].
\end{align}
Thus, for each $(\mu,\nu) \neq (0,0)$, the squared magnitude is the same. Thus, the sum of them yields
\begin{align}
    \operatorname{CCNR}(\rho_\mathcal S) = \sum_{\mu,\nu} \vert \lambda_{\mu\nu} \vert = 1 + (d^2-1)\frac{\sqrt{k}}{\vert \mathcal S\vert}.
\end{align}
\end{proof}

Note that integer solutions to the diophantine equation does not guarantee existence of a corresponding grid. For example, for $d=5$, we checked exhaustively all grids with $\vert \mathcal S\vert=9$, and no solution exists. For $d=6$, however, solutions do exist and we display one example in Fig.~\ref{fig:pattern6}.

Unfortunately, though, none of them is arranged in such a way that the $k=9$ phases cancel for all $\Delta$. This does not guarantee a positive partial transpose, and indeed, all of the solutions (which we were barely able to check) are not PPT. Therefore, we conjecture that no PPT entangled $k$-displacement homogeneous dichotomous Bell diagonal states exist for $d_A=d_B>4$. 

Finally, let us turn to the case of unequal dimensions $d_A < d_B$. While these states are by construction diagonal in the generalized Bell basis, their Bloch representation fails to yield a direct operator Schmidt decomposition, making the evaluation of criteria like CCNR more cumbersome. Nevertheless, we can numerically optimize over the set of dichotomous Bell diagonal states, demanding that they are PPT, but maximizing their violation w.r.t.~CCNR. 
For local dimensions $d_A=3$ and $d_B=4$, we obtain -- up to cyclic permutations of the rows and columns -- two different states, depicted in figures~\ref{fig:pattern34A} and~\ref{fig:pattern34B}.

For $d_A=4$ and $d_B=6$, we obtain two different states (again up to cyclic permutations of the rows and columns). The first one is displayed in Fig.~\ref{fig:pattern46A} and given by the probability matrix
\begin{equation}\label{eq:bds46a}
    P^{(1)} = 
    \frac{1}{10}
    \begin{pmatrix} 
    1 & 0 & 1 & 0 & 0 & 0\\
    1 & 1 & 1 & 0 & 1 & 0\\
    1 & 0 & 1 & 0 & 0 & 0\\
    0 & 0 & 0 & 1 & 0 & 1
    \end{pmatrix},
\end{equation}
and the other one is displayed in Fig.~\ref{fig:pattern46B} and defined by
\begin{equation}\label{eq:bds46b}
    P^{(2)} = 
    \frac{1}{10}
    \begin{pmatrix} 
    1 & 0 & 0 & 0 & 1 & 0\\
    1 & 1 & 1 & 1 & 1 & 0\\
    1 & 0 & 0 & 0 & 1 & 0\\
    0 & 0 & 0 & 0 & 0 & 1
    \end{pmatrix}.
\end{equation}
Both exhibit a CCNR value of
\begin{align}
    \operatorname{CCNR}(\rho_{P^{(1,2)}}) \approx \sqrt{4\cdot6}+0.554.
\end{align}

Exhaustively testing all dichotomous BDS of unequal local dimensions $d_A \leq 4$ and $d_B\leq 6$ yields no further bound entangled states that are detected by CCNR.
In summary, we have constructed examples of bound entangled Bell diagonal states in dimensions $3\times 4$, $4\times 4$ and $4\times 6$.

\subsection{PT-invariant Bell diagonal states}

We now try to simply demand $\rho = \rho^{T_A}$ for the case $d_A=d_B=d$. Using the Bloch representation of $\rho_P$ we obtain
\begin{align}
    \rho_P^{T_A} = \frac 1{d^2} \sum_{\mu,\nu} \lambda_{\mu\nu}\omega^{-\mu\cdot \nu} X^{-\mu} Z^\nu \otimes X^{\mu}Z^{-\nu}.
\end{align}
Comparing this with the original state yields $\lambda_{\mu\nu} = 0$ whenever $\mu \not\equiv -\mu~(\text{mod } d)$. Otherwise, $\lambda_{\mu\nu}\omega^{-\mu\cdot \nu}=\lambda_{\mu\nu}$. The first condition is met by $\mu = 0$ or, if $d$ is even, $\mu = \frac d2$. If $\mu=0$, the second condition is met automatically. If $\mu = \frac d2$, the second condition implies that $\nu \equiv 0~(\text{mod }2)$.

Note that if $d$ is odd, only the $d$ entries in the first row of the $\lambda$ matrix can be non-zero. However, since $\vert \lambda_{\mu\nu} \vert \leq 1$, such states can never violate the CCNR criterion and we are not able to detect their entanglement with this criterion. Thus, we concentrate on the case of even $d$. In this case, we have $d+\frac d2$ non-vanishing entries in the first and $\frac d2$-th row. However, numerical maximization of CCNR over this class of states for $d=4$ and $d=6$ does not yield any state violating the CCNR criterion, making it unlikely that such detectably bound entangled states exist.

Note that for $d_A=2$ and $d_B$ arbitrary, no entangled states with $\rho=\rho^{T_A}$ exist, see~\cite{kraus2000separability}. We did not further investigate the case $d_A\neq d_B$.

\section{SSC witnesses for non-Hermitian operator bases}\label{sec:sscwitnessfornon}

While the CCNR criterion is known to work well in the case of equal dimensions, it fails to detect many entangled states, especially if the local dimensions are not equal. Furthermore, its application usually requires a full state tomography, which is costly in terms of the numbers of required measurements. In practical setups, the tool of entanglement witnesses is often used instead \cite{guhne2009entanglement}. A witness $W$ is an observable such that
\begin{itemize}
    \item $\trace(W\sigma)\geq 0 \text{ for all separable states }\sigma,$
    \item $\trace(W\rho) < 0 \text{ for at least one entangled state }\rho$.
\end{itemize}

To remedy the shortcomings of the CCNR criterion, in Ref.~\cite{class-of-ews}, the authors provided a strong family of separability criteria and sketched the construction of entanglement witness arising from it. However, throughout their paper, they assumed Hermitian operator bases on the local Hilbert spaces. Since we are dealing with Bell diagonal states, which are canonically defined in the non-Hermitian Heisenberg-Weyl basis, we show how to get rid of the latter restriction. This will then allow us to apply a generalization of their results to our BDS.\par
The remaining part of this section is organized as follows: In A, we show how to construct said entanglement witnesses. Except for the slight generalization, the construction will mostly follow Ref.~\cite{class-of-ews}. In B, we show how to analytically optimize said witnesses. Finally, in C, we suggest a numerical optimization procedure that only relies on a previously fixed number of different local measurements, thus avoiding resource consuming full quantum state tomography in practice.
\subsection{General construction of the witnesses}
In the following, we denote by $\{B_i^S\}_{i=0}^{d_S^2-1}$
an operator basis on subsystem $S \in \{A,B\}$, subject only to the restrictions
\begin{equation}\label{eq:orthogonality}
\begin{aligned}
    B_0^S = \one_S, &&
    \trace{\left({B_i^S}^\dagger B_j^S\right)} = d_S \delta_{ij}.
\end{aligned}
\end{equation}
Note that this definition implies that the operators $B_i^S$ are
traceless for $i \neq 0$.
A basis of operators acting on the joint Hilbert space 
$\mathcal{H}$ can be constructed via the pairwise tensor
products of the local basis operators. We define the
\emph{correlation matrix} of a given density operator $\rho$
acting on $\mathcal{H}$ as the complex $d_A^2 \times d_B^2 $ matrix $C$
with the entries
\begin{equation}
    c_{ij} = \trace{\left(\left(B_i^A \otimes
    B_j^B\right)^\dagger \rho\right)}.
\end{equation}
It follows that the density operator $\rho$ can be written as
\begin{equation}
    \rho = \frac{1}{d_A d_B}
    \sum_{i=0}^{d_A^2-1}\sum_{j=0}^{d_B^2-1}
    c_{ij} B_i^A \otimes B_j^B,
\end{equation}
and $\trace{\rho} = 1$ implies that $c_{00}=1$.
We employ the notation
$\|X\|_1 = \trace{\sqrt{X^\dagger X}}$
for the trace norm.
The criterion from Ref.~\cite{class-of-ews} reads:
\begin{theorem}[SSC Bound \cite{class-of-ews}]\label{thm:ssc}
Let $\rho$ be a separable density operator with correlation 
matrix $C$. Then, it holds that
\begin{equation}\label{eq:thm}
    \tracenorm{D_x C D_y} \le \sqrt{d_A-1+x^2}\sqrt{d_B-1+y^2}
\end{equation}
for all $x,y\geq0$, where 
$D_x = \operatorname{diag}(x, 1, \hdots, 1)$ is a 
$d_A^2 \times d_A^2$ matrix, and likewise
$D_y = \operatorname{diag}(y, 1, \hdots, 1)$ is a
$d_B^2 \times d_B^2$ matrix.\footnote{Note that we use a different normalization compared to Ref.~\cite{class-of-ews}.}
\end{theorem}
Note that for $x=y=0$, the de~Vicente criterion \cite{vicente2007separability, de2008further} and for $x=y=1$, the CCNR
criterion \cite{chen2002matrix, rudolph2005further} is recovered. For Bell diagonal states with $d_A = d_B$, the criterion is strongest when choosing $x=y$, and the dependence on this parameter cancels out.

It turns out that the requirement that the basis operators be
Hermitian may be dropped (simply adapt the original proof by allowing for complex valued entries in the correlation matrix). For convenience, we introduce the
functions
\begin{align}
    R(x,y) &:= \sqrt{d_A-1+x^2}\sqrt{d_B-1+y^2} \textup{ and}\\ \label{eq:violation}
    g(x,y) &:= R(x,y) - \tracenorm{D_x C D_y},
\end{align}
such that a state is detected to be entangled by \eqref{eq:thm} if and only if $g(x,y) < 0$ for some $x,y$.
As sketched in \cite{class-of-ews}, this bound allows
for the construction of entanglement witnesses due to the
variational characterization of the trace norm: Let $X$ be a
complex $m\times n$ matrix (we assume $m\le n$). Then, it holds that
\begin{equation}\label{eq:tr-var}
    \tracenorm{X} = \max_{U \in \mathcal{I}(m\times n)} 
    |\trace{\left(U^\dagger X \right)}|,
\end{equation}
where $\mathcal{I}(m\times n)$ denotes the set of isometries of appropriate dimensions, i.e.,~$m\times n$ matrices that satisfy
$U U^\dagger= \one_m$. A proof for the case $m=n$
can be found in \cite{wilde2013quantum}, and the generalization
is straightforward.\par The optimization above can be
carried out over the real part of $\trace{(U^\dagger X)}$ since the
corresponding complex phase can be absorbed into the isometry.
For the construction of the entanglement witnesses, let us fix
operator bases satisfying \eqref{eq:orthogonality}. Now, we 
combine Thm.~\ref{thm:ssc} with \eqref{eq:tr-var}, such that we obtain
for separable states:
\begin{align}
    0 &\le R(x,y) - \tracenorm{D_x C D_y}\\
      &=R(x,y) - \max_{U \in \mathcal{I}(m\times n)} \real \trace \left(D_y U^\dagger D_x C \right)\\
    &= R(x,y) + \min_{U \in \mathcal{I}(m\times n)} 
    \real\trace{\left( D_y U^\dagger D_x C \right)}.     \label{eq:witn-var}
\end{align}
In order to write the right-hand side as an expectation value, we define $\tilde{U} = D_x U D_y$ for an arbitrary isometry $U$ and make use of the normalization of $\rho$, i.e., $c_{00} = 1$, to bound Eq.~(\ref{eq:witn-var}) by
\begin{align}
     0 \le c_{00}R(x,y) + \frac{1}{2}\sum_{i,j}\left(
    c_{ij}\tilde{U}_{ij}^* + c_{ij}^* \tilde{U}_{ij}\right).
    \label{eq:ineq}
\end{align}
Let us consider a generic entanglement
witness $W$. As it is Hermitian, it can be decomposed as $W = \frac{\tilde W + \tilde W ^\dagger}{2}$ for some operator $\tilde W$. Expanding it in the operator basis yields
\begin{equation}\label{eq:witn-dec}
    W = \frac{1}{2}\sum_{i,j}\left(w_{ij} B_i^A\otimes
    B_j^B + 
    w_{ij}^* {B_i^A}^\dagger \otimes {B_j^B}^\dagger \right).
\end{equation}
Exploiting Hermiticity of $\rho$, we find
\begin{equation}\label{eq:witn-structure}
    \trace{\left(W \rho \right)} = \frac{1}{2}\sum_{i,j}
    \left(  c_{ij} w_{ij}^* + c_{ij}^* w_{ij} \right).
\end{equation}
Comparing this to \eqref{eq:ineq} yields the coefficients:
\begin{equation}
\begin{aligned}\label{eq:witn-coeff}
    w_{00}&=xy U_{00} + R(x,y), \textup{ and } w_{i0}=yU_{i0},\\ 
    w_{0j}&=xU_{0j}, w_{ij}=U_{ij} \textup{ for }i,j>0.
\end{aligned}
\end{equation}
Together with \eqref{eq:witn-dec}, this gives an entanglement witness: If $\trace(W\rho)<0$, then $\rho$ is entangled. Note
that for Hermitian operator bases, one recovers the result from
\cite{class-of-ews}.\footnote{Note that, apart from the different normalization, there seem to be additional erroneous factors of $\sqrt{d_A}$ and $\sqrt{d_B}$ in the definitions of $w^{0\beta}$ and $w^{\alpha0}$ in  Ref.~\cite{class-of-ews}.}
\subsection{Analytical optimization of the witnesses}
Given some entangled state $\rho$, we would now like to find the parameters $x,y$ and $U$ characterizing the witness such that its detection power for $\rho$ is optimal. Since the witnesses follow from the criterion in Thm.~\ref{thm:ssc}, which does not depend on the isometry $U$ (recall that $U$ only entered when we constructed the witnesses, which are generally weaker than the criterion itself), we first have to optimize over $x$ and $y$.
This raises the question what optimality means in this context.
The interpretation of the minimum of the function $g$ in Eq.~(\ref{eq:violation}) is unclear, as the norm of the resulting witness depends on $x$ and $y$.
Therefore, we suggest that
the optimal parameter values of $x$ and $y$
correspond to the maximal \emph{noise robustness}:
We choose the parameters such that the given state mixed
with as much white noise as possible is still detected
by the criterion. For $\varepsilon\in [0,1]$, we therefore define
\begin{equation}
    \rho^\varepsilon := (1-\varepsilon)\rho + \varepsilon
    \frac{\one_A\otimes\one_B}{d_Ad_B},
\end{equation}
and call the corresponding SSC violation $g_\varepsilon(x,y)$
(see \eqref{eq:violation}). Optimal values
$(x,y)_*$ for $(x, y)$ are the ones that yield the largest $\varepsilon$ with
\begin{equation}
    \min_{x,y} g_\varepsilon (x,y) < 0.
\end{equation}
There might be multiple such values, in fact the example we study later provides
a one-dimensional curve of optimal values.
We define the \emph{noise threshold} to be the largest such $\varepsilon$
for given $x,y$ and denote it by $\varepsilon_\textrm{max}(x,y)$, such that
\begin{equation}
    (x,y)_*
    \in\underset{(x,y)}{\operatorname{argmax}}\,\varepsilon_\textrm{max}(x,y).
\end{equation}
This two-dimensional optimization can be carried out numerically or analytically
(depending on the structure of $\rho$). Having found optimal values for
$x$ and $y$, all that remains is the optimization over the isometry $U$.
To tackle this, we notice that the trace norm of a generic $m\times n$ matrix $X$
can be expressed in terms of its singular value decomposition
$X = R\varSigma S^\dagger$, where $R$ and $S$ are unitary $m \times m$ and $n \times n$ matrices, respectively, and $\varSigma$ is a $m \times n$ dimensional matrix with real, non-negative entries on the diagonal. Then, assuming $m \le n$, we have
\begin{equation}
    \tracenorm{X} = \trace \begin{pmatrix}
    \varSigma \\ \mathbf{0}\end{pmatrix}
    = \trace \begin{pmatrix} R^\dagger X S\\ \mathbf{0} \end{pmatrix}
    = \trace \left(S \begin{pmatrix} R^\dagger \\ \mathbf{0} \end{pmatrix} X \right),
\end{equation}
where $\mathbf{0}$ denotes the matrix consisting of zeros and is
of dimension $(m-n)\times m$. Applying this to the modified
correlation matrix, i.e., imposing
$D_x C D_y = R \varSigma S^\dagger$, we see that the isometry
minimizing \eqref{eq:witn-var} is given by
\begin{equation}
    U = -(R, \mathbf{0}) S^\dagger,
\end{equation}
where $\mathbf{0}$ is a $d_A^2 \times (d_B^2 - d_A^2)$ matrix
consisting of zeroes. Note that together with \eqref{eq:witn-dec} and
\eqref{eq:witn-coeff}, this yields the optimal witness from the SSC
family of entanglement witnesses for a given quantum state; and the optimized witnesses are, in fact, equivalent to the criterion from Thm.~\ref{thm:ssc}.\par
\subsection{Numerical optimization of the witnesses}
Let us fix some detectable quantum state $\rho$ with corresponding 
correlation matrix $C$.
The results so far rely on the full knowledge of the underlying
density operator. However, we would like to detect entanglement
in a way that requires measuring fewer than $d_A^2 \times d_B^2$
local basis operators. To achieve this, we note that the trace
norm can also be characterized as follows:
\begin{equation}
    \tracenorm{X} = \max_{\|U\|_\infty \le 1} \real\trace
    \left(U^\dagger X \right),
\end{equation}
where $\|U\|_\infty$ is the operator norm and evaluates to the largest
singular value of $U$. We can now make an ansatz for a witness in analogy to
\eqref{eq:witn-dec} and \eqref{eq:witn-coeff}, and numerically
optimize over the $w_{ij}$ with the constraints that
\begin{enumerate}
    \item the corresponding matrix $U$ fulfills
    $\|U\|_\infty \le 1$, and
    \item only $\ell$ entries of $U$ are non-vanishing.
\end{enumerate}
After fixing values $(x,y)$, the resulting expression needs to be optimized over $\ell$ entries of $U$,
which is generally an $\ell$-dimensional complex optimization problem
with non-linear constraints. Finally, we require that the witness
must be Hermitian, and therefore, the optimization boils down to
an $\ell$-dimensional \emph{real} procedure. Let us denote by $\mathcal P _\ell$ the set of possible parameters $(x,y)$ that can detect $\rho$ with $\ell$ local measurements other than the one corresponding to $B_0^A\otimes B_0^B = \one_A\otimes\one_B$ (since this is just the normalization of $\rho$), according to the aforementioned construction of entanglement witnesses. By $\mathcal P $, we denote the full admissable parameter space, i.e., the set of parameter values $(x,y)$ that allow for detection of $\rho$ by directly using Thm.~\ref{thm:ssc}. Clearly, this defines a filtration
\begin{equation}
    \emptyset = \mathcal{P}_0 \subseteq \mathcal{P}_1\subseteq\dots\subseteq \mathcal{P}_{d_A^2 d_B^2 - 1} = \mathcal{P}
\end{equation}
of the full admissable parameter space.
For the numerical implementation, it is feasible to choose those entries in $U$ to be non-vanishing that correspond to the largest entries (by absolute value) in $D_xCD_y$. We denote by $\tilde{\mathcal{P}}_\ell$ the corresponding parameter spaces. It can easily be verified that this construction yields a new filtration
\begin{equation}
    \emptyset = \tilde{\mathcal{P}}_0 \subseteq \tilde{\mathcal{P}}_1\subseteq\dots\subseteq \tilde{\mathcal{P}}_{d_A^2 d_B^2 - 1} = \mathcal{P}.
\end{equation}
\begin{figure}
    \centering
    \includegraphics[width=.4\textwidth]{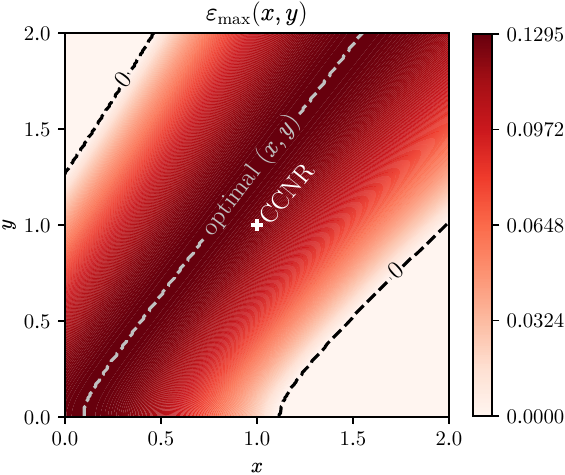}
    \caption{Noise threshold $\varepsilon_\textrm{max}$ for the state $\rho_1$ from Eq.~(\ref{eq:bds46a}) for different values $x$ and $y$ in the criterion in Thm.~\ref{thm:ssc}. Outside of the region enclosed by the dashed black lines, no entanglement is detected. The grey line indicates
    pairs $(x,y)$ that yield the maximal noise threshold $\operatorname{argmax}_{(x,y)}\varepsilon_\textrm{max}(x,y) \approx 0.1295$. The CCNR criterion at $(x,y)=(1,1)$ 
        (marked as a white cross) and the de~Vicente criterion at $(x,y)=(0,0)$ do not belong to this one dimensional curve. For further details, see main text.}
    \label{fig:noise}
\end{figure}
\section{SSC witnesses for Bell diagonal states of unequal local dimension}\label{sec:sscwitnessbds}

Having generalized the entanglement criterion and the witness construction to non-Hermitian bases, 
we are applying it to classes of states which are most naturally defined in such bases, namely Bell diagonal states.
If $d_A = d_B$, the
witness imposed in \eqref{eq:witn-structure} and \eqref{eq:witn-coeff} can be optimized analytically. Then, one finds that $x=y$ is required for optimality, such that the dependence on these parameters cancels out. In that case, one
obtains
\begin{equation}
    W=d\one\otimes\one - \sum_{\mu,\nu} \frac{\lambda_{\mu\nu}}{|\lambda_{\mu\nu}|}
    X^\mu Z^\nu \otimes X^\mu Z^{-\nu},
\end{equation}
which coincides exactly with the known CCNR witnesses
(see, e.g., \cite{gittsovich2010quantum}). However, in the case of $d_A \neq d_B$, there exist Bell diagonal states which are neither detectable by de~Vicente, nor by CCNR, but still are bound entangled and detected as such by different values of $x$ and $y$.\par
To that end, we set $d_A = 4$ and $d_B = 6$ and consider the bound entangled state given by the probability matrix $P^{(1)}$ in Eq.~(\ref{eq:bds46a}).
Let $\rho_1 := \rho_{P^{(1)}}$ denote the corresponding Bell mixture and
\begin{equation}
    \rho_1^\varepsilon = (1-\varepsilon) \rho_1 + 
    \varepsilon\frac{\one_4\otimes\one_6}{24}
\end{equation}
be its noisy version with $\varepsilon \in [0,1]$.
We plot the noise robustness for each choice of $x$ and $y$ in Fig.~\ref{fig:noise}. The state is detected by the SSC criterion for $\varepsilon \lesssim 0.1295$,
and with $\varepsilon = 0.129$, the criterion gives a value of $-6.3\times 10^{-4}$,
whereas the state is neither detected by the de~Vicente criterion nor by the
CCNR criterion. The matrix representation of the corresponding entanglement witness is given in the Appendix. In order to estimate the value of $\varepsilon$ where $\rho_1^\varepsilon$ becomes separable, we used the algorithm from Ref.~\cite{kampermann2012algorithm} and found separable decompositions for $\varepsilon \geq 0.361$.

We see that for $d_A \neq d_B$, the SSC criterion, and therefore
also our constructed witnesses, are stronger than the CCNR
and de~Vicente
criteria (which correspond to $x=y$ as pointed out in \cite{class-of-ews}).

\begin{figure}
    \centering
\includegraphics[width=0.87\columnwidth]{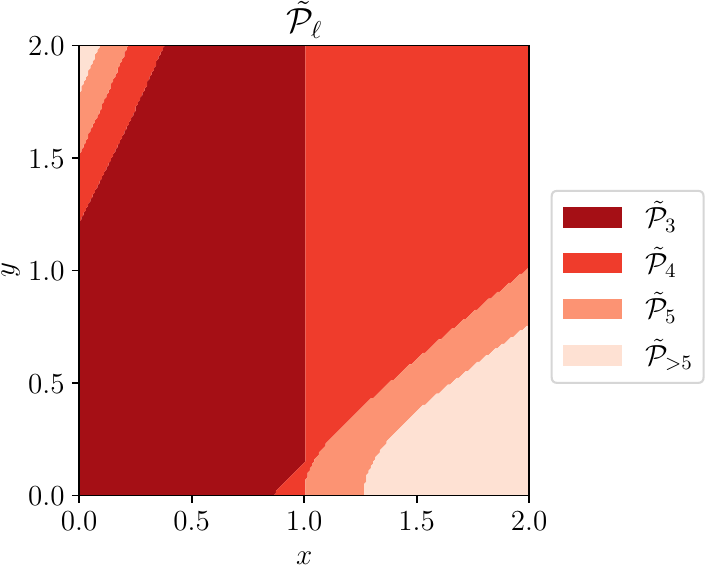}
    \caption{The maximally entangled Bell state $\ket{\phi^{00}}$ in local di\-men\-sions $2 \times 3$ admits suitable entanglement witnesses relying on $\ell\ge 3$ local measurements in the Heisenberg-Weyl basis. The areas $\tilde{\mathcal{P}}_\ell$ correspond to values of $(x,y)$ that require $\ell$ local measurements in order to detect the entanglement. In particular, tuples $(x,y)\in\tilde{\mathcal{P}}_{>5}= \mathcal{P}\setminus\tilde{\mathcal P}_5$ require at least five local measurements. The construction procedure is described in detail in Section~\ref{sec:sscwitnessfornon}.}
    \label{fig:numWitn}
\end{figure}

We now turn to the numerical optimization described in Section~\ref{sec:sscwitnessfornon} and apply the procedure to the state $\ket{\phi^{00}}$  in local dimensions $2\times 3$ (see definition in Eq.~\eqref{eq:bell00}). For suitable parameter values of $x$ and $y$, the numerical optimization yields entanglement witnesses detecting the state whenever we allow for $\ell\ge 3$ local measurements in the Heisenberg-Weyl basis (as well as the trivial one of $\one_A\otimes \one_B$). In Fig.~\ref{fig:numWitn}, we visualize the required number of measurements depending on $x$ and $y$.  The graph shows that we indeed have
\begin{equation}
    \emptyset=\tilde{\mathcal{P}}_0
    =\tilde{\mathcal{P}}_1
    =\tilde{\mathcal{P}}_2
    \subset\tilde{\mathcal{P}}_3\subset\dots\subset \tilde{\mathcal{P}}_5\subseteq\dots\subseteq\mathcal P.
\end{equation}
Furthermore, for the analyzed state $\ket{\psi^{00}}$ in local dimensions $2\times 3$, it holds that $[0,2]\times[0,2]\subset\mathcal P$. Therefore, all parameter values depicted in Fig.~\ref{fig:numWitn} admit witnesses that can detect the state whenever we allow for sufficiently many local measurements. A peculiar feature is the vertical line at $x=1$. It is caused by the fact that we select measurements corresponding to the $\ell$ largest entries of the correlation matrix after scaling it, i.e., of $D_x C D_y$. For the given state, the absolute values of the entries of this matrix are given by
\begin{align}
    (D_x C D_y)_\text{abs} = \frac12\begin{pmatrix}
    2xy & x        & x        & 0   & 0   & 0   & 0   & 0   & 0 \\
    0  & \sqrt{3} & \sqrt{3} & 0   & 0   & 0   & 0   & 0   & 0 \\
    0  & 0          & 0          & 1 & 1 & 1 & 1 & 1 & 1 \\
    0  & 0          & 0          & 1 & 1 & 1 & 1 & 1 & 1
    \end{pmatrix}.
\end{align}
While the value of $y$ only affects the trivial entry corresponding to the identity, the value of $x$ has an effect on which of the entries are largest, with a transition at $x=1$.

To conclude, we visualize the maximal noise robustness for the one-parameter family of dephased states
    $\rho_\theta\coloneqq \Phi_1(\ket{\phi_\theta}\bra{\phi_\theta})$,
introduced in \eqref{eq:phi-theta},
in figure~\ref{fig:noise-theta}. As expected, the noise robustness is maximal when $\rho_\theta$ is Bell-diagonal, i.e., $\theta\in\{0,1\}$. For $\theta=1/2$, no entanglement is detected. Using the PPT criterion, it can be verified that the state $\rho_{\theta=1/2}$ is indeed separable.
\begin{figure}
    \centering
    \includegraphics[width=\columnwidth]{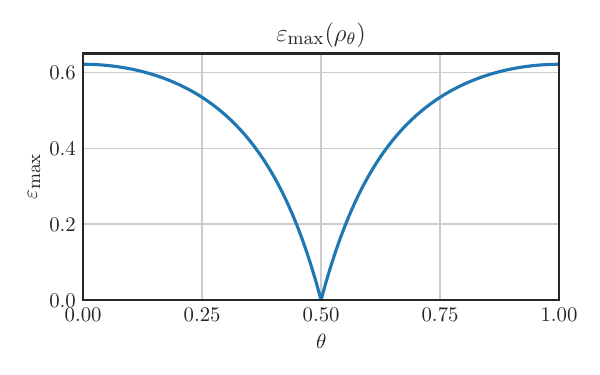}
    \caption{Maximal noise robustness of the one-parameter class of states $\{\rho_\theta\}$, see equation~\eqref{eq:phi-theta} and Section \ref{sec:sscwitnessbds} in the main text. For $\theta= 1/2$, no entanglement is detected.}
    \label{fig:noise-theta}
\end{figure}
\section{Conclusions}
We introduced the notion of Bell diagonal states (BDS) for the case of unequal local Hilbert space dimensions. We proceeded by analyzing their structure, in particular, their alternative characterization via the Fourier picture as well as their entanglement properties. For the latter, we made use of well-known criteria such as the CCNR test to construct bound entangled states with large violation of the respective criteria in different local dimensions.

We then generalized an entanglement criterion found by
Sarbicki, Scala and Chru{\'s}ci{\'n}ski \cite{class-of-ews} (SSC) based on the decomposition of density operators into Hermitian operator bases, to the case of
non-Hermitian operator bases, and demonstrated how to express the arising entanglement witnesses in these bases. This allowed for the application of said criterion to Bell diagonal states which are canonically defined in the non-Hermitian Heisenberg-Weyl basis. We imposed a protocol
which allows for analytical as well as numerical optimization of
entanglement witnesses based on the SSC criterion and showed that
the generated witnesses have the same detection power as the
criterion itself if the involved isometry parameter is chosen appropriately.  Additionally, we provided an optimization
method for (generally weaker) witnesses, which only relies on
a previously fixed number of local measurements, which enhances
applicability in experimental setups. Finally, by applying said criterion to BDS of unequal local dimensions, we showed that the witnesses arising from our construction are generally stronger than the CCNR criterion and allow to detect bound entanglement even when CCNR and the de~Vicente criterion fail to do so.

Nevertheless, many questions remain open for future research. For example, it is known that usual Bell diagonal states arise naturally in the analysis of quantum key distribution schemes, and it would be interesting to investigate the role of generalized Bell diagonal states in asymmetric protocols. Another immediate route forward concerns the entanglement dimensionality of the states as characterized by their Schmidt numbers. As the employed criteria for entanglement have been generalized to detect different Schmidt numbers, it should be possible to generalize the analysis in this paper to that case \cite{wyderka2023probing, liu2023characterizing}.

\section*{Acknowledgments}
\acknowledgments
We thank Otfried Gühne for fruitful discussions. 
The authors acknowledge support by the QuantERA project QuICHE via the German Ministry of Education and Research (BMBF Grant
No. 16KIS1119K). JM acknowledges financial support due to the WE-Heraeus communication program by the DPG.

\bibliographystyle{apsrev4-1}
\bibliography{sources}

\appendix

\section{Numerical construction of entanglement witnesses}

The numerical optimization used to construct entanglement witnesses corresponding to Fig.~\ref{fig:numWitn} was carried out using the \texttt{scipy.optimize.minimize} function in \texttt{Python 3.9.5}. We employed the Heisenberg-Weyl operator basis, such that 
\begin{equation}\begin{aligned}
    B^A_i &:= X_A^{\lfloor i / d_A\rfloor} Z_A^i,
    \text{ and }\\
    B^B_j &:= X_B^{\lfloor j / d_B\rfloor} Z_B^{-j},
\end{aligned}\end{equation}
where the exponents on the RHS may be taken mod $d_A$ and mod $d_B$, respectively.
We took the $\ell$ entries in $U$ corresponding to the largest coefficients (by absolute value) in the correlation matrix of $\ket{\phi^{00}}$ (not counting the normalization $\one_A\otimes \one_B$) to be non-vanishing.\par
We now describe the optimization procedure that we employed to construct the optimal SSC witness to detect the bound entangled state $\rho_{P^{(1)}}$ constructed from the probability matrix in Eq.~(\ref{eq:bds46a}). To that end, we discretized the parameter space of possible values of $0\leq x\leq 2$ and $0\leq y \leq 2$ in steps of 200 each, yielding a total of $40000$ points. For each point, we constructed the optimal isometry as explained in Section~\ref{sec:sscwitnessfornon} and evaluated the noise robustness $\varepsilon_\text{max}(x,y)$ with a divide and conquer scheme. It turns out that the largest noise robustness of $0.1295$ is assumed for a one dimensional subset of values for $x$ and $y$ as displayed in Fig.~\ref{fig:noise}. Identifying the fixed orthonormal basis in $\mathcal{H}$ with the canonical unit vectors in $\mathbb{C}^{d_A}\otimes \mathbb{C}^{d_B}=\mathbb{C}^{d_A \times d_B}$, the corresponding witness is given as a matrix
\begin{align}
    W = \sum_{k,l=0}^{d_Ad_B-1} m_{k,l} \ketbra{k}{l}.
\end{align}
We print the non-vanishing coefficients
(rounded to 4 digits):

\allowdisplaybreaks
\begin{align*} %\begin{gathered} 
m_{0,0} &= m_{7,7} = m_{16,16} = m_{23,23} = 0.2039, \\ 
m_{0,7} &= m_{4,11} = m_{6,5} = m_{11,12} = m_{12,19}\\
&= m_{13,6} = m_{16,23} = m_{17,10} = m_{18,17} = 0.4\mathrm{i}, \\ 
m_{0,14} &= m_{2,16} = m_{7,21} = m_{9,23} = m_{14,0}\\
&= m_{16,2} = m_{21,7} = m_{23,9} = -0.4472, \\ 
m_{0,21} &= m_{2,23} = m_{18,3} = m_{20,5} = -0.4472\mathrm{i}, \\ 
m_{1,1} &= m_{5,5} = m_{6,6} = m_{10,10} = m_{13,13}\\
&= m_{17,17} = m_{18,18} = m_{22,22} = 1.2015, \\ 
m_{1,8} &= m_{8,15} = m_{15,22} = 0.7099\mathrm{i}, \\ 
m_{1,15} &= m_{3,17} = m_{6,20} = m_{8,22} = m_{15,1}\\
&= m_{17,3} = m_{20,6} = m_{22,8} = 0.4472, \\ 
m_{1,22} &= m_{4,19} = -0.6325\mathrm{i}, \\ 
m_{2,2} &= m_{9,9} = m_{14,14} = m_{21,21} = 0.6015,\addtocounter{equation}{1}\tag{\theequation} \\ 
m_{2,9} &= m_{7,14} = m_{9,16} = m_{10,3}\\
&= m_{14,21} = m_{20,13} = 0.555\mathrm{i}, \\ 
m_{3,3} &= m_{8,8} = m_{15,15} = m_{20,20} = 1.3318, \\ 
m_{3,10} &= m_{9,2} = m_{13,20} = m_{14,7}\\
&= m_{16,9} = m_{21,14} = -0.555\mathrm{i}, \\ 
m_{3,18} &= m_{5,20} = m_{21,0} = m_{23,2} = 0.4472\mathrm{i}, \\ 
m_{4,4} &= m_{11,11} = m_{12,12} = m_{19,19} = 1.4689, \\ 
m_{4,12} &= m_{5,13} = m_{10,18} = m_{11,19} = m_{12,4}\\
&= m_{13,5} = m_{18,10} = m_{19,11} = 0.6325, \\ 
m_{5,6} &= m_{6,13} = m_{7,0} = m_{10,17} = m_{11,4}\\
&= m_{12,11} = m_{17,18} = m_{19,12} = m_{23,16} = -0.4\mathrm{i}, \\ 
m_{8,1} &= m_{15,8} = m_{22,15} = -0.7099\mathrm{i}, \\ 
m_{19,4} &= m_{22,1} = 0.6325\mathrm{i}. \\ 
%\end{gathered}
\end{align*}

\end{document}